\documentclass[12pt,a4paper]{article}

\usepackage{amsmath,amssymb}
\usepackage{graphicx}
\usepackage[nosort]{cite}

\usepackage{tikz}
\usetikzlibrary{decorations.markings}
\usepackage{epsfig}  
\usepackage{enumerate}
\usepackage{relsize}
\usepackage{multirow}
\newcommand{\mathsym}[1]{{}}

\usepackage{amsthm,amsbsy,color,multicol}
\usepackage{array,calc}
\usepackage{rotating}
\usepackage{a4wide,bm}
\usepackage[a4paper,text={450pt,650pt},centering]{geometry}
\usepackage{setspace}
\let\oldbfseries=\bfseries
\let\oldmdseries=\mdseries
\let\oldnormalfont=\normalfont
\renewcommand{\bfseries}{\oldbfseries\boldmath}
\renewcommand{\mdseries}{\oldmdseries\unboldmath}
\renewcommand{\normalfont}{\oldnormalfont\unboldmath}

\allowdisplaybreaks[3]

\numberwithin{equation}{section}

\usepackage[font=small,labelfont=bf,width=0.85\textwidth]{caption}

\ifx\hypersetup\sadfkjashdfkxja\newcommand\hypersetup[1]{}\fi

\hypersetup{plainpages=false}
\hypersetup{pdfpagemode=UseNone}
\hypersetup{bookmarksnumbered=true}
\hypersetup{pdfstartview=FitH}
\hypersetup{colorlinks=false}
\hypersetup{citebordercolor={.5 1 .5}}
\hypersetup{urlbordercolor={.5 1 1}}
\hypersetup{linkbordercolor={1 .7 .7}}

\DeclareMathSymbol{\Gamma}{\mathalpha}{letters}{"00}
\DeclareMathSymbol{\Delta}{\mathalpha}{letters}{"01}
\DeclareMathSymbol{\Theta}{\mathalpha}{letters}{"02}
\DeclareMathSymbol{\Lambda}{\mathalpha}{letters}{"03}
\DeclareMathSymbol{\Xi}{\mathalpha}{letters}{"04}
\DeclareMathSymbol{\Pi}{\mathalpha}{letters}{"05}
\DeclareMathSymbol{\Sigma}{\mathalpha}{letters}{"06}
\DeclareMathSymbol{\Upsilon}{\mathalpha}{letters}{"07}
\DeclareMathSymbol{\Phi}{\mathalpha}{letters}{"08}
\DeclareMathSymbol{\Psi}{\mathalpha}{letters}{"09}
\DeclareMathSymbol{\Omega}{\mathalpha}{letters}{"0A}

\newcommand{\gen}[1]{\mathrm{#1}}

\newcommand{\dd}{\mathrm{d}}
\newcommand{\ii}{\mathrm{i}}

\newcommand*\widebar[1]{%
  \hbox{%
    \vbox{%
      \hrule height 0.5pt 
      \kern0.25ex
      \hbox{%
        \kern-0.3em
        \ensuremath{#1}%
        \kern-0.1em
      }%
    }%
  }%
}

\ifx\genfrac\sdflkaj\else\fi

\newcommand{\ket}[1]{\left|#1\right\rangle}      
\newcommand{\bra}[1]{\left\langle #1\right|}     

\newcommand{\alg}[1]{\mathfrak{#1}}

\newcommand{\beq}{\begin{equation}}
\newcommand{\eeq}{\end{equation}}

\def\[{\begin{equation}}
\def\]{\end{equation}}
\def\<{\begin{eqnarray}}
\def\>{\end{eqnarray}}

\newtheorem{mydef}{Definition}

\newtheorem{lemma}{Lemma} 
\newtheorem{remark}{Remark}

\makeatletter
\def\mr@ignsp#1 {\ifx\:#1\@empty\else #1\expandafter\mr@ignsp\fi}%
\newcommand{\multiref}[1]{\begingroup
\xdef\mr@no@sparg{\expandafter\mr@ignsp#1 \: }%
\def\mr@comma{}%
\@for\mr@refs:=\mr@no@sparg\do{\mr@comma\def\mr@comma{,}\ref{\mr@refs}}%
\endgroup}
\makeatother

\newcommand{\hypref}[2]{\ifx\href\asklfhas #2\else\href{#1}{#2}\fi}
\newcommand{\Secref}[1]{Section~\multiref{#1}}

\renewcommand{\eqref}[1]{(\multiref{#1})}

\makeatletter
\newlength{\apb@width}
\newcommand{\autoparbox}[2][c]{\settowidth{\apb@width}{#2}\parbox[#1]{\apb@width}{#2}}

\makeatother

\ifx\href\asklfhas\newcommand{\href}[2]{#2}\fi

\begin{document}

\renewcommand{\thefootnote}{\fnsymbol{footnote}}
\thispagestyle{empty}
\begin{flushright}\footnotesize
ITP-UU-13/25 \\
SPIN-13/18
\end{flushright}
\vspace{1cm}

\begin{center}%
{\Large\bfseries%
\hypersetup{pdftitle={Functional relations and the \\ Yang-Baxter algebra}}%
Functional relations and the \\ Yang-Baxter algebra%
\par} \vspace{2cm}%

\textsc{W. Galleas}\vspace{5mm}%
\hypersetup{pdfauthor={Wellington Galleas}}%

\textit{Institute for Theoretical Physics and Spinoza Institute, \\ Utrecht University, Leuvenlaan 4,
3584 CE Utrecht, \\ The Netherlands}\vspace{3mm}%

\verb+w.galleas@uu.nl+ %

\par\vspace{2.5cm}

\textbf{Abstract}\vspace{7mm}

\begin{minipage}{12.7cm}
Functional equations methods are a fundamental part of the theory of Exactly Solvable Models in
Statistical Mechanics and they are intimately connected with Baxter's concept of commuting transfer matrices.
This concept has culminated in the celebrated Yang-Baxter equation which plays a 
fundamental role for the construction of quantum integrable systems and also for obtaining their exact solution. 
Here I shall discuss a proposal that has been put forward in the past years, in which the Yang-Baxter algebra
is viewed as a source of functional equations describing quantities of physical interest. For instance, this 
method has been successfully applied for the description of the spectrum of open spin chains, partition functions
of elliptic models with domain wall boundaries and scalar product of Bethe vectors. Further applications of 
this method are also discussed.

\hypersetup{pdfkeywords={Yang-Baxter equation, functional relations, partial differential equations}}%
\hypersetup{pdfsubject={}}%
\end{minipage}
\vskip 2cm
{\small December 2013}

\end{center}

\newpage
\pagestyle{empty}
\renewcommand{\thefootnote}{\arabic{footnote}}
\setcounter{footnote}{0}


\section{Introduction}
\label{sec:intro}

Exact solutions have played an important role for the development of physical theories
and assumptions and their contributions can be seen in a variety of contexts. For instance,
Onsager's solution of the two-dimensional Ising model \cite{Onsager_1944} not only
showed that the formalism of Statistical Mechanics was indeed able to describe phase 
transitions, but also unveiled that the critical behavior of the Ising model specific
heat was not included in Landau's theory of critical exponents \cite{Ma}. 
In a different context, the exact solution of the one-dimensional Heisenberg spin chain 
\cite{Bethe_1931} was also of fundamental importance for elucidating the value of the spin
carried by a spin wave \cite{Faddeev_1981}.

Bethe's celebrated solution of the isotropic Heisenberg chain \cite{Bethe_1931} consists of 
a fundamental stone of the modern theory of quantum integrable systems and its influence
can be seen in several areas ranging from Quantum Field Theory \cite{Zamolodchikov_1990, Zarembo_2003}
to Combinatorics \cite{Galleas_Brak}. The hypothesis employed by Bethe for the model wave function is
known nowadays as `Bethe ansatz' and it became a standard tool in the theory of quantum integrable systems.
On the other hand, Onsager's solution of the two-dimensional Ising model was based on Kramers and
Wannier transfer matrix technique \cite{Kramers_1941a , Kramers_1941b} which did not have any 
previous connection with Bethe ansatz. Nevertheless, the transfer matrix technique was later on recognized
as a fundamental ingredient for the establishment of integrability in the sense of Baxter \cite{Baxter_1971}.

Within this scenario Baxter's concept of integrability appeared as an analogous of Liouville's
classical concept where now the transfer matrix was playing the role of generating function
of quantities in involution \cite{Baxter_1971}. More precisely, in Baxter's framework a family of
mutually commuting operators is obtained as a consequence of transfer matrices which commute for different values
of their parameters. In their turn, these commutative transfer matrices are built directly
from solutions of the Yang-Baxter equation. 

The importance of the Yang-Baxter equation was only better understood with the 
proposal of the Quantum Inverse Scattering Method \cite{Sk_1979, Fad_1979}. This method unified
the transfer matrix approach, the Yang-Baxter equation and the Bethe ansatz employed to solve
a variety of one-dimensional quantum many-body systems. Besides that, the Quantum Inverse
Scattering Method, or QISM for short, put in evidence the so called Yang-Baxter algebra which
latter on led to the notion of Quantum Groups \cite{chari1995guide}.

The Yang-Baxter algebra plays a fundamental role within the QISM and it is one of the main
ingredients for the construction of exact eigenvectors of transfer matrices of two-dimensional lattice
systems and hamiltonians of one-dimensional quantum many-body systems.
However, the applications of the Yang-Baxter algebra are not limited to that and alternative ways of exploring the Yang-Baxter
algebra are also known in the literature. For instance, it can be used to build solutions of the Knizhnik-Zamolodchikov 
equation \cite{KZ_1984} in the sense of \cite{Babujian_1993, Babujian_1994}.

More recently, the Yang-Baxter algebra was also shown to be capable of rendering functional equations
describing quantities such as the spectrum of spin chains and partition functions of vertex
models \cite{Galleas_2008, Galleas10, Galleas11, Galleas_2011, Galleas_2012, Galleas_SC}. 
Here we aim to discuss this latter possibility. 

This article is organized as follows. In \Secref{sec:LAT} we introduce definitions which will be relevant 
throughout this paper and also present the lattice systems we shall consider by means of this algebraic-functional 
approach. In \Secref{sec:YBAFZ} we illustrate how the Yang-Baxter algebra can be converted into 
functional equations and, in particular, we derive functional relations describing the partition function
of the elliptic Eight-Vertex-SOS model with domain wall boundaries and scalar products of Bethe vectors.
The solutions of the aforementioned functional equations are also presented in \Secref{sec:YBAFZ},
and in \Secref{sec:PDE} we unveil a family of partial differential equations underlying 
our functional relations. Concluding remarks are then discussed in \Secref{sec:conclusion}.

\section{Yang-Baxter relations and lattice systems}
\label{sec:LAT}

Lattice systems of Statistical Mechanics have a long history and remarkable examples share
the property of being exactly solvable \cite{Baxter_book}. The most prominent examples, such
as the Ising model and Eight-Vertex model, have been solved in two-dimensions and this choice of
dimensionality undoubtedly grants them special properties. For instance, in \cite{Smirnov_2001, Smirnov_pre} 
Smirnov proved that the scaling limit of the critical site percolation on a two-dimensional triangular lattice
is conformally invariant. The importance of this proof can be seen in two ways: from the mathematical perspective Smirnov's
proof introduced the concept of `discrete' harmonic functions. On the other hand, this proof provides
a solid ground for the CFT (Conformal Field Theory) methods employed by Cardy in \cite{Cardy_1992}.

The concept of exact solvability seems to be intrinsically dependent on the method we are employing. 
However, it is nowadays well accepted that Baxter's commuting transfer matrices approach \cite{Baxter_1971}
plays a fundamental role for two-dimensional lattice systems within a variety of methods. For instance, the
requirement of commuting transfer matrices leads us to the Yang-Baxter equation/algebra \cite{Sk_1979, Fad_1979} 
and also their dynamical counterparts \cite{Felder_1994, Felder_1996}. Those algebraic relations constitute
the foundations of the algebraic Bethe ansatz \cite{Fad_1979} and, as firstly demonstrated in \cite{Galleas_2008}, 
they are also able to describe spectral problems in terms of functional equations.

Vertex and Solid-on-Solid models are some important examples of exactly solvable lattice systems
and both of them admit an operatorial description in terms of generators of the Yang-Baxter algebra and its 
dynamical version. Here we will be mainly interested in the so called Eight-Vertex-SOS and Six-Vertex models, and for that
it is enough to present only the dynamical Yang-Baxter equation/algebra while the standard Yang-Baxter relations
will be obtained as a particular limit.

\paragraph{Dynamical Yang-Baxter equation.} Let $\alg{g}$ be a finite dimensional Lie algebra over 
$\mathbb{C}$ and $\alg{h} \subset \alg{g}$ be an abelian Lie subalgebra. Also, let 
$\mathbb{V} = \bigoplus_{\phi \in \alg{h}^{*}} \mathbb{V}[\phi]$ with
$\mathbb{V}[\phi] = \{ v \in \mathbb{V} \; | \;  h v = \phi(h) v \; \mbox{for} \; h \in \alg{h} \}$ 
be a diagonalizable $\alg{h}$-module. Then for $\lambda_j , \gamma, \theta \in \mathbb{C}$ the dynamical
Yang-Baxter equation reads
\<
\label{dyn}
\mathcal{R}_{12}(\lambda_1 - \lambda_2 , \theta - \gamma h_3) \mathcal{R}_{13}(\lambda_1 - \lambda_3 , \theta) \mathcal{R}_{23}(\lambda_2 - \lambda_3 , \theta - \gamma h_1) = \nonumber \\
\mathcal{R}_{23}(\lambda_2 - \lambda_3 , \theta)   \mathcal{R}_{13}(\lambda_1 - \lambda_3 , \theta - \gamma h_2)  \mathcal{R}_{12}(\lambda_1 - \lambda_2 , \theta) \; .
\>
Eq. (\ref{dyn}) is a relation for an operator $\mathcal{R}_{ij} \; : \; \mathbb{C} \times \alg{h}^{*} \mapsto \mbox{End}( \mathbb{V}_i \otimes \mathbb{V}_j )$ 
where $\mathbb{V}_i$ ($i=1,2,3$) are finite dimensional diagonalizable $\alg{h}$-modules.  
In this way we have (\ref{dyn}) defined in $\mbox{End}( \mathbb{V}_1 \otimes \mathbb{V}_2 \otimes \mathbb{V}_3 )$
with tensor products being understood as
\[
\label{tensor}
\mathcal{R}_{12}(\lambda , \theta - \gamma h_3) ( v_1 \otimes v_2 \otimes v_3 ) = (\mathcal{R}_{12}(\lambda , \theta - \gamma \phi) ( v_1 \otimes v_2  )) \otimes v_3 \; .
\]
The term $\phi$ in (\ref{tensor}) corresponds to the weight of $v_3$ while the remaining elements $\mathcal{R}_{13}$ and $\mathcal{R}_{23}$
are computed by analogy. As far as the solutions of (\ref{dyn}) are concerned, it is currently well understood the importance of the elliptic
quantum groups $E_{p, q} [ \alg{g} ]$ for their characterization. Here we shall restrict ourselves to the case $\alg{g} \simeq \alg{sl}(2)$,
and in that case we consider $\alg{h}$ as the $\alg{sl}(2)$ Cartan subalgebra while $\mathbb{V} \cong \mathbb{C}^2$.
Thus $h = \mbox{diag}(1, -1)$ and we have the explicit solution
\begin{align} \label{rmat}
\mathcal{R} (\lambda, \theta) &= \left( \begin{matrix}
a_{+}(\lambda, \theta) & 0 & 0 & 0 \cr 
0 & b_{+}(\lambda, \theta) & c_{+}(\lambda, \theta) & 0 \cr
0 & c_{-}(\lambda, \theta) & b_{-}(\lambda, \theta) & 0 \cr
0 & 0 & 0 & a_{-}(\lambda, \theta) \end{matrix} \right) & 
& \begin{matrix} \quad \quad a_{\pm}(\lambda, \theta) = f(\lambda + \gamma) \\ 
\quad \quad \;\; b_{\pm}(\lambda, \theta) = f(\lambda) \frac{f(\theta \mp \gamma)}{f(\theta)} \\
\quad \quad \;\; c_{\pm}(\lambda, \theta) = f(\gamma) \frac{f(\theta \mp \lambda)}{f(\theta)} \end{matrix} \; . \nonumber \\
\end{align}
The function $f$ in (\ref{rmat}) is essentially a Jacobi Theta-function. More precisely we have
$f(\lambda) = \Theta_1 (\ii \lambda, \tau)/2$, according to the conventions of \cite{Watson}, 
and the dependence of $f$ with the elliptic nome $\tau$ is omitted for convenience.  

\paragraph{Dynamical Yang-Baxter algebra.} Let $\mathbb{V}_{a} \cong \mathbb{V}$,
$\mathbb{V}_{\mathcal{Q}} \cong \mathbb{V}^{\otimes L}$ and consider an operator
$\mathcal{T}_{a} \in \mbox{End}( \mathbb{V}_{a} \otimes \mathbb{V}_{\mathcal{Q}})$
for an arbitrary integer $L$. Then the dynamical Yang-Baxter equation (\ref{dyn}) ensures the associativity of the 
relation
\<
\label{RTT}
\mathcal{R}_{ab}(\lambda_1 - \lambda_2, \theta - \gamma \gen{H}) \mathcal{T}_a (\lambda_1, \theta) \mathcal{T}_b (\lambda_2, \theta - \gamma \hat{h}_a) = \nonumber \\ 
\mathcal{T}_b (\lambda_2, \theta) \mathcal{T}_a (\lambda_1, \theta - \gamma \hat{h}_b ) \mathcal{R}_{ab}(\lambda_1 - \lambda_2, \theta) \; , 
\>
for $\mathcal{R}$-matrices satisfying the weight zero condition, i.e. $[ \mathcal{R} , h\otimes 1 + 1 \otimes h ] = 0$.
As usual we define $h_i \in \mathbb{V}_{\mathcal{Q}}$ as $h$ acting on the $i$-th node of the tensor product space
$\mathbb{V}_{\mathcal{Q}}$ while $\gen{H} = \sum_{i=1}^L h_i$. The generator $\gen{H}$ can be identified with the 
$\alg{sl}(2)$ Cartan generator acting on $\mathbb{V}_{\mathcal{Q}}$ and one can also show that
\[
\label{rep}
\mathcal{T}_{a} (\lambda, \theta) = \mathop{\overrightarrow\prod}\limits_{1 \le i \le L } \mathcal{R}_{a i}(\lambda - \mu_i, \hat{\theta}_i) 
\]
is a representation of (\ref{RTT}) with $\hat{\theta}_i = \theta - \gamma \sum_{k=i+1}^L h_{k}$ and arbitrary parameters $\mu_i \in \mathbb{C}$.
The operator $\mathcal{T}$ is usually denominated dynamical monodromy matrix, or
simply monodromy matrix, and it consists of a matrix in the space $\mathbb{V}_{a}$ whose entries are operators
in the space $\mathbb{V}_{\mathcal{Q}}$. Thus for the $E_{p, q} [ \alg{sl}(2) ]$ solution (\ref{rmat}),
our monodromy matrix can be recasted as
\[
\label{ABCD}
\mathcal{T} (\lambda , \theta) = \left( \begin{matrix}
\mathcal{A}(\lambda , \theta) & \mathcal{B}(\lambda , \theta) \\
\mathcal{C}(\lambda , \theta) & \mathcal{D}(\lambda , \theta) \end{matrix} \right)
\]
where $\mathcal{A}, \mathcal{B}, \mathcal{C}, \mathcal{D} \in \mbox{End}( (\mathbb{C}^2 )^{\otimes L})$. 

\begin{remark} \label{limit} An important limit of (\ref{rmat}) is the limit $p = e^{\ii \pi \tau} \rightarrow 0$
where the elliptic Theta-functions degenerate into trigonometric ones. Furthermore, if we also 
consider the limit $\theta \rightarrow \infty$ we are left with the standard  $\mathcal{R}$-matrix
and Yang-Baxter relations of the six-vertex model \cite{Baxter_book}.
\end{remark}

\paragraph{Integrable lattice systems.} Exact solvability of two-dimensional lattice systems can be achieved
from certain conditions of integrability in analogy to the theory of integrable differential 
equations. This condition is fulfilled by commuting transfer matrices which is assured by local
equivalence transformations satisfied by the model statistical weights \cite{Baxter_book}.
For vertex models we can encode the statistical weight of a given vertex configuration as the entry of a
certain matrix $\mathcal{R}$. In this way the aforementioned equivalence transformation requires this $\mathcal{R}$-matrix
to satisfy the celebrated Yang-Baxter equation. For Solid-on-Solid models, or SOS for short, the condition of integrability 
in the sense of Baxter requires the model statistical weights to satisfy the so called Hexagon identity \cite{integrable_book}.
This condition was shown in \cite{Felder_1994} to be directly related to the dynamical Yang-Baxter equation (\ref{dyn}), 
and in what follows we shall briefly describe some important examples of integrable SOS and vertex lattice systems.

\begin{enumerate}
\item {\bf Eight-Vertex-SOS model with domain wall boundaries:} This model consists of a two-dimensional lattice system defined
on a square lattice. It is built from the juxtaposition of plaquettes where we associate a set of state variables
to the corners of each plaquette in order to characterize its allowed configurations. As far as boundary conditions are concerned,
the case of domain wall boundaries consists of the assumption that the plaquettes at the border are fixed at a particular
configuration. This model has been previously considered in \cite{Rosengren_2008, Pakuliak_2008, WengLi_2009, Razumov_2009a}
and here we shall adopt the conventions of \cite{Galleas_2011, Galleas_2012}. In this way the partition function
of the Eight-Vertex-SOS model with domain wall boundaries can be written as
\[
\label{pf}
\mathcal{Z}_{\theta} = \bra{\bar{0}} \mathop{\overrightarrow\prod}\limits_{1 \le j \le L } \mathcal{B}(\lambda_j, \theta + j \gamma) \ket{0} \; ,
\]
where the operators $\mathcal{B}(\lambda , \theta)$ are defined through the relations (\ref{ABCD}), (\ref{rmat}) and (\ref{rep}).
In their turn the vectors $\ket{0}$ and $\ket{\bar{0}}$ are explicitly given by
\<
\label{states}
\ket{0} = \left( \begin{matrix}
1 \\ 0 \end{matrix} \right)^{\otimes L} \qquad \qquad 
\ket{\bar{0}} = \left( \begin{matrix}
0 \\ 1 \end{matrix} \right)^{\otimes L} \; .
\>

\item {\bf Scalar product of Bethe vectors:} The evaluation of partition functions of lattice systems such as vertex models
with periodic boundary conditions can be conveniently translated into an eigenvalue problem for an operator usually denominated
transfer matrix \cite{Kramers_1941a, Kramers_1941b}. The diagonalization of transfer matrices for integrable vertex models 
can be performed, for instance, through the Quantum Inverse Scattering Method \cite{Fad_1979}. Within that method, Bethe vectors
arise as an ansatz capable of determining transfer matrices exact eigenvectors. In a slightly different context, scalar products of Bethe vectors
can also be regarded as the partition function of a vertex model with special boundary conditions \cite{Korepin82, deGier_Galleas11},
and this is the interpretation we shall pursue here. 
As it was stressed out in Remark \ref{limit}, the standard six-vertex model relations can be obtained from (\ref{rmat})
in a particular limit. In that limit we have $\mathcal{A}(\lambda , \theta) \rightarrow A(\lambda)$, $\mathcal{B}(\lambda , \theta) \rightarrow B(\lambda)$, 
$\mathcal{C}(\lambda , \theta) \rightarrow C(\lambda)$, $\mathcal{D}(\lambda , \theta) \rightarrow D(\lambda)$ and the scalar product
of Bethe vectors $S_n$ then reads
\[
\label{sn}
S_n = \bra{0} \mathop{\overleftarrow\prod}\limits_{1 \le i \le n } C(\lambda_i^{C}) \; \mathop{\overrightarrow\prod}\limits_{1 \le i \le n } B(\lambda_i^{B})  \ket{0} \; .
\]
In formulae (\ref{sn}) we have considered the description employed in \cite{deGier_Galleas11} while the vector $\ket{0}$ has been defined 
in (\ref{states}). 
\end{enumerate}

From (\ref{pf}) and (\ref{sn}) we can see that the above defined partition functions are given as the expected value of a product
of generators of the Yang-Baxter algebra and its dynamical counterpart. In what follows we shall demonstrate how 
algebraic relations (\ref{RTT}) can be explored in order to obtain functional equations determining the aforementioned quantities.

\section{Yang-Baxter algebra and functional relations}
\label{sec:YBAFZ}

In order to illustrate how the Yang-Baxter algebra can be employed to derive functional relations, 
we shall consider the six-vertex model limit of (\ref{dyn}), (\ref{rmat}) and (\ref{RTT}) for the sake
of simplicity. In particular, the relation (\ref{RTT}) then encodes a total of sixteen commutation rules
involving the set of generators $\mathcal{M}(\lambda) = \{A , B , C , D \}(\lambda)$ evaluated at different values of the 
spectral parameter $\lambda$. 

From the perspective of Quantum Field Theory one can regard the operators $A(\lambda)$ and 
$D(\lambda)$ as diagonal fields while $B(\lambda)$ and $C(\lambda)$ plays the role of creation and annihilation fields.
As far as we are concerned with an eigenvalue problem involving the set of generators $\mathcal{M} (\lambda)$, i.e. the diagonalization
of a transfer matrix \cite{Korepin_book, Fad_1979}, the framework of Quantum Field Theory is quite appealing and it seems natural
to build the corresponding eigenvectors as elements of a Fock space. In particular, this approach is encouraged by the structure of the commutation
relations in  (\ref{RTT}). Those commutation rules constitute one of the corner stones of the algebraic Bethe ansatz \cite{Fad_1979} and 
for instance let us single out the following one
\[
\label{AB}
A(\lambda_1) B(\lambda_2) = \frac{a(\lambda_2 - \lambda_1)}{b(\lambda_2 - \lambda_1)} B(\lambda_2) A(\lambda_1) - \frac{c(\lambda_2 - \lambda_1)}{b(\lambda_2 - \lambda_1)}  B(\lambda_1) A(\lambda_2) \; ,
\]
where $a(\lambda) = \sinh{(\lambda + \gamma)}$, $b(\lambda) = \sinh{(\lambda)}$ and $c(\lambda) = \sinh{(\gamma)}$. 
Within the framework of the algebraic Bethe ansatz we usually regard (\ref{AB}) as a relation between a diagonal field
and a creation field.  Nevertheless, this is not the only way one can explore relations of type
(\ref{AB}), and in what follows we shall see they can also be projected as a functional relation.

\begin{mydef} \label{Pi2}
Let $\gen{\pi} : \mbox{End}( \mathbb{V}_{\mathcal{Q}} ) \times \mbox{End}( \mathbb{V}_{\mathcal{Q}} ) \mapsto \mathbb{C}$ be a continuous
and bi-additive map. Due to (\ref{AB}), or more generally (\ref{RTT}), it is convenient to specialize the map $\gen{\pi}$ to $\gen{\pi}_2$ 
defined as
\[
\label{pi2}
\gen{\pi}_2 : \qquad \mathcal{M}(\lambda) \times \mathcal{M}(\mu) \mapsto \mathbb{C}[\lambda^{\pm 1}, \mu^{\pm 1}] \; .
\] 
The $2$-tuple $(\xi_1 , \xi_2) : \xi_1 \in \mathcal{M}(\lambda) , \; \xi_2 \in \mathcal{M}(\mu)$ 
originated from the Cartesian product $\mathcal{M}(\lambda) \times \mathcal{M}(\mu)$ is then simply
understood as the matrix product $\xi_1 \xi_2$. In other words the map $\gen{\pi}_2$ associates a two-variable
complex function to any quadratic form in (\ref{RTT}). 
\end{mydef}

The map $\gen{\pi}_2$ defined in (\ref{pi2}) is able to associate a functional relation to any commutation
rule contained in (\ref{RTT}). For instance, the map (\ref{pi2}) applied on (\ref{AB}) yields the relation
\[
\label{fAB}
b(\lambda_2 - \lambda_1) f(\lambda_1 , \lambda_2) = a(\lambda_2 - \lambda_1) \bar{f} (\lambda_2 , \lambda_1) - c(\lambda_2 - \lambda_1) \bar{f} (\lambda_1 , \lambda_2) 
\]
where $f(\lambda_1 , \lambda_2) = \gen{\pi}_2 ( A(\lambda_1) B(\lambda_2) )$ and  $\bar{f}(\lambda_1 , \lambda_2) = \gen{\pi}_2 ( B(\lambda_1) A(\lambda_2) )$. 

The study of functional equations has a long history, see for instance the monograph \cite{Aczel_book}, and
they play a remarkable role in Statistical Mechanics \cite{Baxter_book} and Conformal Field Theory \cite{Bazhanov_1997}.
Within those contexts they appear intimately related to Baxter's concept of commuting transfer matrices \cite{Baxter_1971} and
among prominent examples we have Baxter's $T-Q$ relation \cite{Baxter_1972}, inversion relation \cite{Stroganov_1982}, analytical
Bethe ansatz \cite{Reshet_1987} and $Y$-system \cite{Kuniba_1994, Kuniba_2011}. Also, it is worth mentioning the quantum 
Knizhnik-Zamolodchikov equation \cite{Frenkel_Reshetikhin} which describes form factors and correlation functions in integrable
field theories \cite{Smirnov_book}. Here we intend to demonstrate that partition functions of integrable lattice models,
such as (\ref{pf}) and (\ref{sn}), can also be described by functional equations. Interestingly, the functional equations
describing those partition functions follow directly from the Yang-Baxter algebra within the lines above discussed.
In order to show that we first need to generalize the Definition \ref{Pi2} in the following way.

\begin{mydef}
Let $n$ be an integer. Then we define the $n$-additive continuous map $\gen{\pi}_n$ as
\[
\label{pin}
\gen{\pi}_n : \qquad \mathcal{M}(\lambda_1) \times \mathcal{M}(\lambda_2) \times \dots \times \mathcal{M}(\lambda_n) \mapsto \mathbb{C}[\lambda_1^{\pm 1} , \lambda_2^{\pm 1} , \dots , \lambda_n^{\pm 1} ] \; .
\]
Similarly to (\ref{pi2}) the $n$-tuple $(\xi_1 , \xi_2 , \dots , \xi_n) : \xi_i \in \mathcal{M}(\lambda_i)$ is understood
as the non-commutative product $\mathop{\overrightarrow\prod}\limits_{1 \le i \le n } \xi_i$. In other words,
the map $\gen{\pi}_n$ associates a $n$-variable complex function to a product of $n$ generators of the Yang-Baxter algebra.
\end{mydef}

\paragraph{Realization of $\gen{\pi}_n$.} A simple choice of realization of $\gen{\pi}_n$ is
the scalar product with vectors $\ket{\psi} , \ket{\psi'} \in \mathbb{V}_{\mathcal{Q}}$.
More precisely, we can readily see that
\[
\label{pinR}
\gen{\pi}_n ( \mathcal{F} ) = \bra{\psi'} \mathcal{F} \ket{\psi} 
\] 
is a realization of $\gen{\pi}_n$ for any element $\mathcal{F} \in \mathcal{M}(\lambda_1) \times \mathcal{M}(\lambda_2) \times \dots \times \mathcal{M}(\lambda_n)$.
At this stage $\ket{\psi}$ and $\ket{\psi'}$ are arbitrary vectors but it will become clear that particular choices can render
interesting functional equations describing quantities such as (\ref{pf}) and (\ref{sn}).

\subsection{Functional equation for $\mathcal{Z}_{\theta}$}
\label{sec:FEZ}

In this section we aim to show how a functional equation for the partition function $\mathcal{Z}_{\theta}$
can be derived within the lines discussed in \Secref{sec:YBAFZ}. For that the first step is to find
appropriate vectors $\ket{\psi}$ and $\ket{\psi'}$ in order to employ the realization (\ref{pinR}).
In addition to that it is also important to consider suitable elements $\mathcal{F}$ such that
we end up with an equation capable of determining the desired partition function. 
Although there is no precise recipe for selecting $\ket{\psi}$, $\ket{\psi'}$ and $\mathcal{F}$,
we shall see that some properties of the elements entering in the definitions (\ref{pf}) and (\ref{states}) 
can help us to sort that out.

\paragraph{Highest weight vectors.} The vectors $\ket{0}$ and $\ket{\bar{0}}$ defined in (\ref{states})
are $\alg{sl}(2)$ highest and lowest weight vectors. They satisfy the following properties: 

\begin{align}
\label{Raction}
\mathcal{A}(\lambda , \theta) \ket{\bar{0}} & = \frac{f(\theta - \gamma)}{f(\theta + (L-1)\gamma)} \prod_{j=1}^{L} f(\lambda - \mu_j) \ket{\bar{0}} & \mathcal{A}(\lambda , \theta) \ket{0} & = \prod_{j=1}^{L} f(\lambda - \mu_j + \gamma) \ket{0} \nonumber \\
\mathcal{D}(\lambda , \theta) \ket{0} & = \frac{f(\theta + \gamma)}{f(\theta - (L-1)\gamma)} \prod_{j=1}^{L} f(\lambda - \mu_j) \ket{0} & \mathcal{D}(\lambda , \theta) \ket{\bar{0}} & = \prod_{j=1}^{L} f(\lambda - \mu_j + \gamma) \ket{\bar{0}} \nonumber \\
\mathcal{C}(\lambda , \theta) \ket{0} & = 0 & B(\lambda) \ket{\bar{0}} & = 0 
\end{align}
\begin{align}
\label{Laction}
\bra{\bar{0}} \mathcal{A}(\lambda , \theta)  & = \frac{f(\theta - \gamma)}{f(\theta + (L-1)\gamma)} \prod_{j=1}^{L} f(\lambda - \mu_j) \bra{\bar{0}} & \bra{0} \mathcal{A}(\lambda , \theta)  & = \prod_{j=1}^{L} f(\lambda - \mu_j + \gamma) \bra{0} \nonumber \\
\bra{0} \mathcal{D}(\lambda , \theta)  & = \frac{f(\theta + \gamma)}{f(\theta - (L-1)\gamma)} \prod_{j=1}^{L} f(\lambda - \mu_j) \bra{0} & \bra{\bar{0}} \mathcal{D}(\lambda , \theta)  & = \prod_{j=1}^{L} f(\lambda - \mu_j + \gamma) \bra{\bar{0}} \nonumber \\
\bra{\bar{0}} \mathcal{C}(\lambda , \theta)  & = 0 & \bra{0} B(\lambda)  & = 0 
\end{align}
The expressions (\ref{Raction}) and (\ref{Laction}) follow from definitions (\ref{rep}), (\ref{rmat})
and the highest/lowest weight property of $\ket{0}$ and $\ket{\bar{0}}$.
\medskip

\paragraph{Yang-Baxter relations of degree $n$.} The relations arising from the 
dynamical Yang-Baxter algebra (\ref{RTT}) involve the set of generators 
$\mathcal{M}(\lambda , \theta) = \{ \mathcal{A}, \mathcal{B} , \mathcal{C} , \mathcal{D} \}(\lambda , \theta)$
and $\gen{H}$. Although some entries of (\ref{RTT}) contain products of the form 
$\mathcal{M}(\lambda_1 , \theta_1) \times \mathcal{M}(\lambda_2 , \theta_2) \times f(\gen{H})$, there are still commutation
rules with all terms in  $\mathcal{M}(\lambda_1 , \theta_1) \times \mathcal{M}(\lambda_2 , \theta_2)$. Those are the relations
that will be explored here and among them we have the following ones
\<
\label{BB}
\mathcal{B}(\lambda_1, \theta) \mathcal{B}(\lambda_2, \theta +\gamma) &=& \mathcal{B}(\lambda_2, \theta) \mathcal{B}(\lambda_1, \theta +\gamma) \nonumber \\
\mathcal{A}(\lambda_1, \theta + \gamma) \mathcal{B}(\lambda_2, \theta) &=& \frac{f(\lambda_2 - \lambda_1 +\gamma)}{f(\lambda_2 - \lambda_1)} \frac{f(\theta +\gamma)}{f(\theta + 2 \gamma)} \mathcal{B}(\lambda_2, \theta + \gamma) \mathcal{A}(\lambda_1, \theta + 2\gamma) \nonumber \\
&-& \frac{f(\theta + \gamma - \lambda_2 + \lambda_1)}{f(\lambda_2 - \lambda_1)} \frac{f(\gamma)}{f(\theta + 2\gamma)} \mathcal{B}(\lambda_1, \theta + \gamma) \mathcal{A}(\lambda_2, \theta + 2\gamma)  \; . \nonumber \\
\>
Both expressions in (\ref{BB}) are quadratic and their repeated use is able to provide relations of degree $n$ for
a subset of elements $\mathcal{F} \subset \mathcal{W}_n = \mathcal{M}(\lambda_0, \theta_0) \times \mathcal{M}(\lambda_1 , \theta_1) \times \dots \times \mathcal{M}(\lambda_{n-1} , \theta_{n-1})$.
More precisely, the iteration of (\ref{BB}) yields the following relation of degree $n+1$, 
\<
\label{ABn}
&& \mathcal{A}(\lambda_0, \theta + \gamma) Y_{\theta - \gamma} (\lambda_1 , \dots , \lambda_n) = \nonumber \\
&& \frac{f(\theta + \gamma)}{f(\theta + (n+1) \gamma)} \prod_{j=1}^{n} \frac{f(\lambda_j - \lambda_0 + \gamma)}{f(\lambda_j - \lambda_0)} Y_{\theta} (\lambda_1 , \dots , \lambda_n) \mathcal{A}(\lambda_0 , \theta + (n+1)\gamma) \nonumber \\
&& - \sum_{i=1}^{n} \frac{f(\theta + \gamma - \lambda_i + \lambda_0)}{f(\theta + (n+1)\gamma)} \frac{f(\gamma)}{f(\lambda_i - \lambda_0)} \prod_{\stackrel{j=1}{j \neq i}}^{n} \frac{f(\lambda_j - \lambda_i + \gamma)}{f(\lambda_j - \lambda_i)} \times \nonumber \\
&& \qquad \quad Y_{\theta} (\lambda_0 , \lambda_1 , \dots , \lambda_{i-1} , \lambda_{i+1} , \dots , \lambda_n) \mathcal{A}(\lambda_i , \theta + (n+1)\gamma) \; ,
\>
where $Y_{\theta} (\lambda_1 , \dots , \lambda_n) = \mathop{\overrightarrow\prod}\limits_{1 \le j \le n } B(\lambda_j, \theta + j\gamma)$.

\paragraph{Building $\gen{\pi}_n$ for $\mathcal{Z}_{\theta}$.} The relations (\ref{Raction}) and (\ref{Laction}) suggest the 
prescription $\ket{\psi} = \ket{0}$ and $\ket{\psi'} = \ket{\bar{0}}$. As we shall see this particular
choice of $\gen{\pi}_n$ is able to generate a functional equation for the partition function $\mathcal{Z}_{\theta}$ from the algebraic 
relation (\ref{ABn}).

\medskip 
Taking into account the discussion of \Secref{sec:YBAFZ} we then set $n = L$ in the relation (\ref{ABn}).  
Next we consider the action of the map $\gen{\pi}_{L+1}$ on (\ref{ABn}) and by doing so we find only terms of the form
\[
\label{pr}
\gen{\pi}_{L+1} ( \mathcal{A}(\lambda_0, \theta + \gamma) Y_{\theta - \gamma} (\lambda_1 , \dots , \lambda_L) )
\]
and
\[
\label{pr1}
\gen{\pi}_{L+1} ( Y_{\theta} (\lambda_0 , \lambda_1 , \dots , \lambda_{i-1} , \lambda_{i+1} , \dots , \lambda_L) \mathcal{A}(\lambda_i , \theta + (L+1)\gamma)  ) \;\; .  
\]
Interestingly, the relations (\ref{Raction}) and (\ref{Laction}) obtained as a consequence of the highest/lowest
weight property of $\ket{0}$ and $\ket{\bar{0}}$ give rise to a map $\gen{\pi}_{L+1} \mapsto \gen{\pi}_{L}$.
More precisely we have
\<
\label{PLL}
\gen{\pi}_{L+1}( \mathcal{A}(\lambda_0, \theta + \gamma) Y_{\theta - \gamma} (\lambda_1 , \dots , \lambda_L) ) &=& \frac{f(\theta)}{f(\theta + L \gamma)} \prod_{j=1}^{L} f(\lambda_0 - \mu_j) \;\; \gen{\pi}_{L}( Y_{\theta - \gamma} (\lambda_1 , \dots , \lambda_L) ) \nonumber \\
\>
and
\<
\label{PLL1}
\gen{\pi}_{L+1}( Y_{\theta} (\lambda_0 , \lambda_1 , \dots , \lambda_{i-1} , \lambda_{i+1} , \dots , \lambda_L) \mathcal{A}(\lambda_i , \theta + (L+1)\gamma)  ) = \nonumber \\
\prod_{j=1}^{L} f(\lambda_i - \mu_j) \;\; \gen{\pi}_{L}( Y_{\theta} (\lambda_0 , \lambda_1 , \dots , \lambda_{i-1} , \lambda_{i+1} , \dots , \lambda_L) ) \; .
\>

The partition function (\ref{pf}) can now be promptly identified with $\gen{\pi}_{L}( Y_{\theta} (\lambda_1 , \dots , \lambda_L ) )$. 
In this way the action of $\gen{\pi}_{L+1}$ on (\ref{ABn}), in addition to the relations (\ref{PLL}) and (\ref{PLL1}), leaves us 
with the following functional equation for $\mathcal{Z}_{\theta}$,
\<
\label{FZ}
M_0 \; \mathcal{Z}_{\theta - \gamma}(\lambda_1, \dots , \lambda_L) + \sum_{i=0}^{L} N_i \; \mathcal{Z}_{\theta}(\lambda_0, \dots , \lambda_{i-1}, \lambda_{i+1}, \dots, \lambda_L) = 0  \; .  
\>
The structure of the coefficients $M_0$ and $N_i$ is a direct consequence of the dynamical Yang-Baxter algebra relations (\ref{RTT})
and the highest weight properties (\ref{Raction}) and (\ref{Laction}). For convenience we shall postpone presenting their explicit
form. It is also worth remarking that (\ref{FZ}) made its first appearance in \cite{Galleas_2012} and due to the fact that the operators $\mathcal{B}$ satisfy 
the relation $\mathcal{B}(\lambda_1, \theta) \mathcal{B}(\lambda_2, \theta +\gamma) = \mathcal{B}(\lambda_2, \theta) \mathcal{B}(\lambda_1, \theta +\gamma)$,
as given by (\ref{BB}), the ordering of the arguments of $\mathcal{Z}_{\theta}$ in (\ref{FZ}) is indeed arbitrary. 
This property leads us to the following lemma.

\begin{lemma}[Symmetric function] The partition function $\mathcal{Z}_{\theta}$ is symmetric
under the permutation of its variables, i.e. 
\[
\label{sym}
\mathcal{Z}_{\theta} ( \dots , \lambda_i , \dots , \lambda_j , \dots ) = \mathcal{Z}_{\theta} ( \dots , \lambda_j , \dots , \lambda_i , \dots ) \; .
\]
\end{lemma}
\begin{proof}
This property follows directly from the commutation relation (\ref{BB}) or from the 
functional equation (\ref{FZ}) as demonstrated in \cite{Galleas_2012}.
\end{proof}
\begin{remark} \label{X}
Due to (\ref{sym}) we can safely employ the notation $\mathcal{Z}_{\theta} (\lambda_1 , \dots , \lambda_L) = \mathcal{Z}_{\theta} ( X^{1,L} ) $
where $X^{i,j} = \{ \lambda_k \; : \; i \le k \le j \}$. 
\end{remark}

Taking into account Remark \ref{X}, it is also convenient to introduce the set $X^{i,j}_k = X^{i,j} \backslash \{ \lambda_k \}$
in such a way that (\ref{FZ}) can be simply recasted as
\[
\label{FX}
M_0 \; \mathcal{Z}_{\theta - \gamma} (X^{1,L}) + \sum_{i=1}^{L} N_i \; \mathcal{Z}_{\theta} ( X^{0,L}_i )  = 0 \; .
\]
In their turn the coefficients $M_0$ and $N_i$ explicitly read
\<
\label{coeff}
M_0 &=& \frac{f(\theta)}{f(\theta + L\gamma)} \prod_{j=1}^{L} f(\lambda_0 - \mu_j) \nonumber \\
N_i &=& - \frac{f(\theta + \gamma + \lambda_0 - \lambda_i)}{f(\theta + (L+1)\gamma)} \frac{f(\gamma)}{f(\lambda_0 - \lambda_i + \gamma)} \prod_{j=1}^{L} f(\lambda_i - \mu_j + \gamma) \prod_{\lambda \in X^{0,L}_i} \frac{f(\lambda - \lambda_i + \gamma)}{f(\lambda - \lambda_i)}  \; . \nonumber \\
\>

Some further remarks are important at this stage. The reader familiar with the theory of Knizhnik-Zamolodchikov (KZ) equations
can notice that (\ref{FX}) exhibits a structure which resembles that of the classical KZ equation \cite{KZ_1984}.
For instance, the first term of (\ref{FX}) consists of the partition function with shifted variable
$\theta$ which could be thought of as an analogous of the derivative.
The second term of (\ref{FX}) consists of a linear combination of partition functions with a given variable $\lambda_i$ in
the argument being replaced by a variable $\lambda_0$. This variable replacement can be regarded as the action
of an operator which can be considered as a sort of `hamiltonian'.
Although KZ equations are vector equations while here we are dealing with a scalar equation, we can see that
both terms of (\ref{FX}) have a counterpart in the KZ theory. Furthermore, we will find that the solution of
(\ref{FX}) also resembles solutions of KZ-like equations \cite{Varchenko_book}. 

\paragraph{Solution.} The functional relation (\ref{FX}) consists of an equation for the partition function  
$\mathcal{Z}_{\theta}(X^{1,L})$ over the set of variables $X^{0,L}$. Thus we have one more variable than is required 
to describe the partition function itself. This feature is typical of functional equations such as the d'Alembert
equation \cite{Aczel_book}, but had not appeared previously in the functional relations describing Exactly Solvable Models
to the best of our knowledge. This extra variable can be set at will in order to help us with the resolution of (\ref{FX}), and this approach
is the basis of the method considered in \cite{Galleas_2011} and \cite{Galleas_2012}.
Moreover, Eq. (\ref{FX}) also exhibits some special properties which give us insightful information 
in order to determine its solution. These properties are as follows:

\begin{itemize}
\item {\bf Scale invariance}: Eq. (\ref{FX}) is invariant under the symmetry transformation 
$\mathcal{Z}_{\theta} ( X^{1,L} ) \mapsto \alpha \mathcal{Z}_{\theta} ( X^{1,L} )$ where $\alpha \in \mathbb{C}$
is independent of $\theta$ and $\lambda_j$. This property tells us that (\ref{FX}) is only able
to determine the partition function up to an overall multiplicative factor independent of
the variables $\theta$ and $\lambda_j$. In this way the full determination of the partition function
will require we know the precise value of $\mathcal{Z}_{\theta}$ for a particular choice of the aforementioned variables.

\item {\bf Linearity}: Eq. (\ref{FX}) is linear and as usual this implies that 
if $\mathcal{Z}_{\theta}^{(1)}$ and $\mathcal{Z}_{\theta}^{(2)}$ are solutions, then the linear combination
$\mathcal{Z}_{\theta} = \mathcal{Z}_{\theta}^{(1)} + \mathcal{Z}_{\theta}^{(2)}$ is also a solution. As far as the determination
of (\ref{pf}) is concerned, this property is telling us we need to establish an uniqueness 
criterium in order to characterize the desired partition function. 
\end{itemize} 

Here we do not intend to give a detailed description of the method employed to solve
(\ref{FX}). Nevertheless, we can still comment on how the properties of scale invariance
and linearity have been worked out. In \cite{Galleas_2012} we have shown that the partition function (\ref{pf})
can be explicitly computed in the limit $(\theta, \lambda_j) \rightarrow \infty$. This result
can then be used to completely fix the overall multiplicative factor which is not constrained
by (\ref{FX}). Concerning the issue raised by the linearity of (\ref{FX}), we have also shown in
\cite{Galleas_2012} that the desired solution consists of a higher order Theta-function \cite{Weber}.
As such, it is uniquely characterized by its zeroes up to an overall factor.

As a matter of fact, unveiling special zeroes of $\mathcal{Z}_{\theta}$ plays an important role
for the resolution of (\ref{FX}) and we have found the following solution in \cite{Galleas_2012},
\<
\label{sol}
&& \mathcal{Z}_{\theta} (X^{1, L}) = [f'(0) f(\gamma)]^L  \nonumber \\
&& \oint \dots \oint \prod_{j=1}^{L} \frac{\dd w_j}{2 \ii \pi} \frac{\prod_{j>i}^{L} f(w_j - w_i + \gamma) f(w_j - w_i)}{\prod_{i,j=1}^{L} f(w_i - \lambda_j)} \prod_{j=1}^{L} \frac{f(\theta + j \gamma - w_{j} + \mu_j)}{f(\theta + j \gamma)} \nonumber \\
&& \qquad \qquad \times \prod_{j<i}^{L}  f(\mu_j - w_i) \prod_{j>i}^{L}  f(w_i - \mu_j + \gamma) \; .
\>
Formulae (\ref{sol}) is given in terms of a multiple contour integral whose integration
contour for each variable $w_j$ encloses all variables in the set $X^{1, L}$.
It is also worth remarking that similar multiple contour integrals also emerge as solutions of 
the KZ equation \cite{Varchenko_book} and its $q$-deformed version \cite{Etingof_book}.

\subsection{Functional equation for $S_n$}
\label{sec:FES}

The partition function (\ref{pf}) is not the only quantity which satisfy a functional
equation such as the one described in \Secref{sec:FEZ}. Similar equations can also be derived
for scalar products of Bethe vectors as we shall demonstrate.
Although the case of domain wall boundary conditions was introduced in \cite{Korepin82} as a building block of
scalar products, here we shall not follow that approach but instead consider scalar products defined by (\ref{sn})
as an independent quantity.

The derivation of (\ref{FZ}) required two main ingredients: the construction of a suitable
realization of $\gen{\pi}_n$ and the derivation of appropriate Yang-Baxter relations of degree $n$.
In order to apply the same methodology for scalar products $S_n$ we then first need to consider 
the commutation relations from (\ref{RTT}) in the six-vertex model limit as discussed in Remark \ref{limit}.
In what follows we present the ones that will be required,
\<
\label{ABC}
A(\lambda_1) B(\lambda_2) &=& \frac{a(\lambda_2 - \lambda_1)}{b(\lambda_2 - \lambda_1)} B(\lambda_2) A(\lambda_1) - \frac{c(\lambda_2 - \lambda_1)}{b(\lambda_2 - \lambda_1)} B(\lambda_1) A(\lambda_2) \nonumber \\
C(\lambda_1) A(\lambda_2) &=& \frac{a(\lambda_1 - \lambda_2)}{b(\lambda_1 - \lambda_2)} A(\lambda_2) C(\lambda_1) - \frac{c(\lambda_1 - \lambda_2)}{b(\lambda_1 - \lambda_2)} A(\lambda_1) C(\lambda_2) 
\>
\<
\label{DBC}
D(\lambda_1) B(\lambda_2) &=& \frac{a(\lambda_1 - \lambda_2)}{b(\lambda_1 - \lambda_2)} B(\lambda_2) D(\lambda_1) - \frac{c(\lambda_1 - \lambda_2)}{b(\lambda_1 - \lambda_2)} B(\lambda_1) D(\lambda_2) \nonumber \\
C(\lambda_1) D(\lambda_2) &=& \frac{a(\lambda_2 - \lambda_1)}{b(\lambda_2 - \lambda_1)} D(\lambda_2) C(\lambda_1) - \frac{c(\lambda_2 - \lambda_1)}{b(\lambda_2 - \lambda_1)} D(\lambda_1) C(\lambda_2)
\>
\<
\label{BC}
B(\lambda) B(\mu) &=& B(\mu) B(\lambda) \nonumber \\
C(\lambda) C(\mu) &=& C(\mu) C(\lambda) \; .
\>

\paragraph{Yang-Baxter relation of degree $n$.} The relations (\ref{ABC}), (\ref{DBC}) and (\ref{BC})
are a subset of the commutation rules contained in (\ref{RTT}) in the six-vertex model limit and, as
such, they are relations of degree $2$ according to the discussion of \Secref{sec:FEZ}. 
In order to describe the scalar product $S_n$ we need an appropriate relation of degree $n$ which can be 
obtained by the repeated use of (\ref{ABC})-(\ref{BC}). 
The direct inspection of the commutation relations (\ref{ABC})-(\ref{BC}) suggests us to consider the 
quantity 
\[
\label{TA}
T_{A} = \mathop{\overleftarrow\prod}\limits_{1 \le i \le n } C(\lambda_i^{C}) \;  A(\lambda_0) \mathop{\overrightarrow\prod}\limits_{1 \le i \le n } B(\lambda_i^{B}) \; ,
\]
which can be analyzed in at least two different ways through the relations (\ref{ABC}) and (\ref{BC}).
Here we shall restrict our discussion to the following ways of evaluating $T_A$. Firstly, we can move the operator  
$A(\lambda_0)$ in (\ref{TA}) all the way to the right through all the string of operators $B$ with the help of the first relation in (\ref{ABC}).
Alternatively, we can also move the operator $A(\lambda_0)$ to the left by making use of the second relation in (\ref{ABC}).
These two different ways of evaluating the same quantity yields the following Yang-Baxter relation of order $2n+1$,
\<
\label{TAY}
&& \prod_{i=1}^{n} \frac{a( \lambda_i^C - \lambda_0 )}{b( \lambda_i^C - \lambda_0 )} A(\lambda_0) \mathop{\overleftarrow\prod}\limits_{1 \le i \le n } C(\lambda_i^{C})  \mathop{\overrightarrow\prod}\limits_{1 \le i \le n } B(\lambda_i^{B}) \nonumber \\
&& - \sum_{i=1}^n \frac{c( \lambda_i^C - \lambda_0 )}{b( \lambda_i^C - \lambda_0 )} \prod_{\stackrel{j=1}{j \neq i}}^{n} \frac{a( \lambda_j^C - \lambda_i^C )}{b( \lambda_j^C - \lambda_i^C )} A(\lambda_i^C) \mathop{\overleftarrow\prod}\limits_{\stackrel{0 \le j \le n}{j \neq i}} C(\lambda_j^{C}) \mathop{\overrightarrow\prod}\limits_{1 \le j \le n } B(\lambda_j^{B})  = \nonumber \\
&& \prod_{i=1}^{n} \frac{a( \lambda_i^B - \lambda_0 )}{b( \lambda_i^B - \lambda_0 )} \mathop{\overleftarrow\prod}\limits_{1 \le i \le n } C(\lambda_i^{C})  \mathop{\overrightarrow\prod}\limits_{1 \le i \le n } B(\lambda_i^{B}) \; A(\lambda_0)  \nonumber \\
&& - \sum_{i=1}^n \frac{c( \lambda_i^B - \lambda_0 )}{b( \lambda_i^B - \lambda_0 )} \prod_{\stackrel{j=1}{j \neq i}}^{n} \frac{a( \lambda_j^B - \lambda_i^B )}{b( \lambda_j^B - \lambda_i^B )}  \mathop{\overleftarrow\prod}\limits_{1 \le j \le n} C(\lambda_j^{C}) \mathop{\overrightarrow\prod}\limits_{\stackrel{0 \le j \le n}{j \neq i}} B(\lambda_j^{B}) \; A(\lambda_i^B)  \; .
\>
It is important to stress here that the derivation of (\ref{TAY}) also makes explicit use of the relations (\ref{BC}).
The expression (\ref{TAY}) will be left at rest for a while and we shall focus on another quantity. For instance, 
instead of (\ref{TA}) we could have performed the same analysis starting with 
\[
\label{TD}
T_{D} = \mathop{\overleftarrow\prod}\limits_{1 \le i \le n } C(\lambda_i^{C}) \;  D(\lambda_0) \mathop{\overrightarrow\prod}\limits_{1 \le i \le n } B(\lambda_i^{B}) \; .
\]
In that case we need to consider the relations (\ref{DBC}) and (\ref{BC}), and we end up with the following identity,
\<
\label{TDY}
&& \prod_{i=1}^{n} \frac{a( \lambda_0 - \lambda_i^C  )}{b( \lambda_0 - \lambda_i^C )} D(\lambda_0) \mathop{\overleftarrow\prod}\limits_{1 \le i \le n } C(\lambda_i^{C})  \mathop{\overrightarrow\prod}\limits_{1 \le i \le n } B(\lambda_i^{B}) \nonumber \\
&& - \sum_{i=1}^n \frac{c( \lambda_0 - \lambda_i^C  )}{b( \lambda_0 - \lambda_i^C  )} \prod_{\stackrel{j=1}{j \neq i}}^{n} \frac{a( \lambda_i^C - \lambda_j^C )}{b( \lambda_i^C - \lambda_j^C )} D(\lambda_i^C) \mathop{\overleftarrow\prod}\limits_{\stackrel{0 \le j \le n}{j \neq i}} C(\lambda_j^{C}) \mathop{\overrightarrow\prod}\limits_{1 \le j \le n } B(\lambda_j^{B})  = \nonumber \\
&& \prod_{i=1}^{n} \frac{a( \lambda_0 - \lambda_i^B  )}{b( \lambda_0 - \lambda_i^B )} \mathop{\overleftarrow\prod}\limits_{1 \le i \le n } C(\lambda_i^{C})  \mathop{\overrightarrow\prod}\limits_{1 \le i \le n } B(\lambda_i^{B}) \; D(\lambda_0)  \nonumber \\
&& - \sum_{i=1}^n \frac{c( \lambda_0 - \lambda_i^B   )}{b( \lambda_0 - \lambda_i^B )} \prod_{\stackrel{j=1}{j \neq i}}^{n} \frac{a( \lambda_i^B - \lambda_j^B )}{b( \lambda_i^B - \lambda_j^B )}  \mathop{\overleftarrow\prod}\limits_{1 \le j \le n} C(\lambda_j^{C}) \mathop{\overrightarrow\prod}\limits_{\stackrel{0 \le j \le n}{j \neq i}} B(\lambda_j^{B}) \; D(\lambda_i^B)  \; .
\>
Both expressions (\ref{TAY}) and (\ref{TDY}) consist of Yang-Baxter relations of order $2n +1$ and they can be converted into 
functional equations for $S_n$ with a proper choice of $\gen{\pi}_{2n+1}$.

\paragraph{The map $\gen{\pi}_{m}$ for $S_n$.} Taking into account the relations (\ref{Raction}), (\ref{Laction}),
(\ref{TAY}) and (\ref{TDY}) we can readily see that the choice $\ket{\psi} = \ket{\psi'} = \ket{0}$ for the realization
(\ref{pinR}) gives rise to functional relations for the scalar product $S_n$. In order to see that we only need to
apply the map $\gen{\pi}_{2n+1}$ on (\ref{TAY}) and (\ref{TDY}). By doing so we only find terms of the form
\begin{align}
&\gen{\pi}_{2n+1} ( A(v_0^A) \mathop{\overleftarrow\prod}\limits_{1 \le i \le n } C(v_i^{C})  \mathop{\overrightarrow\prod}\limits_{1 \le i \le n } B(v_i^{B}) )& 
\quad & , \quad \gen{\pi}_{2n+1} ( \mathop{\overleftarrow\prod}\limits_{1 \le i \le n } C(v_i^{C})  \mathop{\overrightarrow\prod}\limits_{1 \le i \le n } B(v_i^{B}) \; A(\bar{v}_0^A) ) \nonumber \\
&\gen{\pi}_{2n+1} ( D(v_0^D) \mathop{\overleftarrow\prod}\limits_{1 \le i \le n } C(v_i^{C})  \mathop{\overrightarrow\prod}\limits_{1 \le i \le n } B(v_i^{B}) )&
\quad & , \quad \gen{\pi}_{2n+1} ( \mathop{\overleftarrow\prod}\limits_{1 \le i \le n } C(v_i^{C})  \mathop{\overrightarrow\prod}\limits_{1 \le i \le n } B(v_i^{B}) \; D(\bar{v}_0^D) )  
\end{align}
with parameters $v_0^A , v_0^D , v_i^C \in \{ \lambda_0 , \lambda_1^C , \dots , \lambda_n^C \}$ 
and $\bar{v}_0^A , \bar{v}_0^D , v_i^B \in \{ \lambda_0 , \lambda_1^B , \dots , \lambda_n^B \}$.

Similarly to the case discussed in \Secref{sec:FEZ}, here we also have a map $\gen{\pi}_{2n +1} \mapsto \gen{\pi}_{2n}$ 
induced by the highest weight property of $\ket{0}$. More precisely, from (\ref{Raction}) we obtain 
\<
\label{PNN1}
\gen{\pi}_{2n+1} ( \mathop{\overleftarrow\prod}\limits_{1 \le i \le n } C(v_i^{C})  \mathop{\overrightarrow\prod}\limits_{1 \le i \le n } B(v_i^{B}) \; A(\bar{v}_0^A) ) &&= \nonumber \\
\prod_{j=1}^{L} a(\bar{v}_0^A - \mu_j) && \gen{\pi}_{2n} ( \mathop{\overleftarrow\prod}\limits_{1 \le i \le n } C(v_i^{C})  \mathop{\overrightarrow\prod}\limits_{1 \le i \le n } B(v_i^{B}) ) \nonumber \\
\gen{\pi}_{2n+1} ( \mathop{\overleftarrow\prod}\limits_{1 \le i \le n } C(v_i^{C})  \mathop{\overrightarrow\prod}\limits_{1 \le i \le n } B(v_i^{B}) \; D(\bar{v}_0^A) ) &&= \nonumber \\
\prod_{j=1}^{L} b(\bar{v}_0^A - \mu_j) && \gen{\pi}_{2n} ( \mathop{\overleftarrow\prod}\limits_{1 \le i \le n } C(v_i^{C})  \mathop{\overrightarrow\prod}\limits_{1 \le i \le n } B(v_i^{B}) ) \; . \nonumber \\
\>
On the other hand, the property (\ref{Laction}) yields
\<
\label{PNN2}
\gen{\pi}_{2n+1} ( A(v_0^A) \mathop{\overleftarrow\prod}\limits_{1 \le i \le n } C(v_i^{C})  \mathop{\overrightarrow\prod}\limits_{1 \le i \le n } B(v_i^{B}) ) &&= \nonumber \\
\prod_{j=1}^{L} a(v_0^A - \mu_j) && \gen{\pi}_{2n} ( \mathop{\overleftarrow\prod}\limits_{1 \le i \le n } C(v_i^{C})  \mathop{\overrightarrow\prod}\limits_{1 \le i \le n } B(v_i^{B}) ) \nonumber \\
\gen{\pi}_{2n+1} ( D(v_0^D) \mathop{\overleftarrow\prod}\limits_{1 \le i \le n } C(v_i^{C})  \mathop{\overrightarrow\prod}\limits_{1 \le i \le n } B(v_i^{B}) ) &&= \nonumber \\
\prod_{j=1}^{L} b(v_0^D - \mu_j) && \gen{\pi}_{2n} ( \mathop{\overleftarrow\prod}\limits_{1 \le i \le n } C(v_i^{C})  \mathop{\overrightarrow\prod}\limits_{1 \le i \le n } B(v_i^{B}) ) \; . \nonumber \\
\> 
The terms $\gen{\pi}_{2n} ( \mathop{\overleftarrow\prod}\limits_{1 \le i \le n } C(v_i^{C})  \mathop{\overrightarrow\prod}\limits_{1 \le i \le n } B(v_i^{B}) )$
can now be identified with the scalar products $S_n$ as defined in (\ref{sn}).
In this way the map $\gen{\pi}_m$ above discussed, together with the relations (\ref{TAY}), (\ref{TDY}), (\ref{PNN1}) and (\ref{PNN2}),
leaves us with the following functional equations,
\<
\label{SNAD}
J_0  S_n ( X^{1,n} |  Y^{1,n} ) &+& \sum_{i=1}^n K_i^{(B)} S_n ( X^{0,n}_i |  Y^{1,n} ) + \sum_{i=1}^n K_i^{(C)} S_n ( X^{1,n} |  Y^{0,n}_i ) = 0 \nonumber \\
\tilde{J}_0  S_n ( X^{1,n} |  Y^{1,n} ) &+& \sum_{i=1}^n \tilde{K}_i^{(B)} S_n ( X^{0,n}_i |  Y^{1,n} ) + \sum_{i=1}^n \tilde{K}_i^{(C)} S_n ( X^{1,n} |  Y^{0,n}_i ) = 0  \; .
\>
In their turn the coefficients appearing in (\ref{SNAD}) can be conveniently written as
\<
\label{coeffA}
J_0 &=& \prod_{j=1}^{L} a(\lambda_0 - \mu_j) \left[  \prod_{i=1}^{n} \frac{a(\lambda_i^{C} - \lambda_0)}{b(\lambda_i^{C} - \lambda_0)} - \prod_{i=1}^{n} \frac{a(\lambda_i^{B} - \lambda_0)}{b(\lambda_i^{B} - \lambda_0)} \right] \nonumber \\
K_i^{(B, C)} &=& \alpha_{B, C} \frac{c(\lambda_i^{B,C} - \lambda_0)}{b(\lambda_i^{B,C} - \lambda_0)} \prod_{j=1}^{L} a(\lambda_i^{B, C} - \mu_j) \prod_{\stackrel{j=1}{j \neq i}}^{n} \frac{a(\lambda_j^{B, C} - \lambda_i^{B, C})}{b(\lambda_j^{B, C} - \lambda_i^{B, C})} \; ,
\>
and 
\<
\label{coeffD}
\widetilde{J}_0 &=& \prod_{j=1}^{L} b(\lambda_0 - \mu_j) \left[  \prod_{i=1}^{n} \frac{a(\lambda_0 - \lambda_i^{C})}{b(\lambda_0 - \lambda_i^{C})} - \prod_{i=1}^{n} \frac{a(\lambda_0 - \lambda_i^{B})}{b(\lambda_0 - \lambda_i^{B})} \right] \nonumber \\
\widetilde{K}_i^{(B, C)} &=& \alpha_{B, C} \frac{c(\lambda_0 - \lambda_i^{B,C})}{b(\lambda_0 - \lambda_i^{B,C})} \prod_{j=1}^{L} b(\lambda_i^{B, C} - \mu_j) \prod_{\stackrel{j=1}{j \neq i}}^{n} \frac{a(\lambda_i^{B, C} - \lambda_j^{B, C})}{b(\lambda_i^{B, C} - \lambda_j^{B, C})} \; .
\>
where $\alpha_B =1$ and $\alpha_C = - 1$.

The Remark \ref{X} of \Secref{sec:FEZ} can be immediately extended to the case of scalar products,
and in (\ref{SNAD}) we have employed the notation
\[
S_n ( \lambda_1^{B}, \dots , \lambda_n^{B} | \lambda_1^{C}, \dots , \lambda_n^{C} ) = S_n ( X^{1,n} |  Y^{1,n} )
\]
where $X^{i,j} = \{ \lambda_k^{B} \; : \; i \leq k \leq j \}$ and $Y^{i,j} = \{ \lambda_k^{C} \; : \; i \leq k \leq j \}$.
Moreover, we have also considered the definitions $X^{i,j}_k = X^{i,j} \backslash \{ \lambda_k^{B} \}$ and $Y^{i,j}_k = Y^{i,j} \backslash \{ \lambda_k^{C} \}$.
This possibility is granted by the following lemma.
\begin{lemma}[Doubly symmetric function]
The scalar product $S_n ( \lambda_1^{B}, \dots , \lambda_n^{B} | \lambda_1^{C}, \dots , \lambda_n^{C} )$
is symmetric under the permutation of variables $\lambda_i^B \leftrightarrow \lambda_j^B$ and  
$\lambda_i^C \leftrightarrow \lambda_j^C$ performed independently. More precisely we have
\<
S_n ( \dots , \lambda_i^B , \dots , \lambda_j^B , \dots | \lambda_1^{C}, \dots , \lambda_n^{C} ) = S_n ( \dots , \lambda_j^B , \dots , \lambda_i^B , \dots | \lambda_1^{C}, \dots , \lambda_n^{C} ) 
\>
and
\<
S_n (\lambda_1^{C}, \dots , \lambda_n^{C} | \dots , \lambda_i^B , \dots , \lambda_j^B , \dots ) = S_n ( \lambda_1^{C}, \dots , \lambda_n^{C} | \dots , \lambda_j^B , \dots , \lambda_i^B , \dots ) \; .
\>
\end{lemma}
\begin{proof}
The proof follows directly from the commutation rules (\ref{BC}). Alternatively, one can demonstrate it from the analysis of (\ref{SNAD})
as performed in \cite{Galleas_SC}.
\end{proof}

\paragraph{Solution.} The same discussion of \Secref{sec:FEZ} concerning the resolution of the functional equation (\ref{FZ})
is also valid for the set of equations (\ref{SNAD}). For instance, we can readily see that each equation in (\ref{SNAD})
is scale invariant and linear. Nevertheless, there is one main difference concerning (\ref{SNAD}) which is the fact
that here we have two equations instead of only one. This might suggest that one of the equations is redundant but the direct inspection
of our system of equations reveals that this is not the case. In fact, the scalar products we are interested consist
of certain polynomials and the use of a polynomial ansatz for solving (\ref{SNAD}) shows that only one equation is not
able to fix all the coefficients. On the other hand, the simultaneous resolution of both equations indeed fix the coefficients 
up to an overall multiplicative factor. 

Although the process of solving the system of equations (\ref{SNAD}) is more involving than
solving the single equation (\ref{FZ}), the same methodology still applies. The solution of (\ref{SNAD})
was firstly obtained in \cite{Galleas_SC} and here we restrict ourselves to presenting only the final expression. 
The scalar product $S_n$ is then given by,
\<
\label{LUA}
S_n ( X^{1,n} | Y^{1,n} ) =  \oint \dots \oint \prod_{i=1}^{n} \frac{\dd w_i}{2 \ii \pi} \frac{\dd \bar{w}_i}{2 \ii \pi} \frac{H(w_1 , \dots , w_n | \bar{w}_1 , \dots , \bar{w}_n)}{ \prod_{i,j=1}^{n} b(w_i - \lambda_j^C) b(\bar{w}_i - \lambda_j^B)} \; ,
\>
where the function $H$ explicitly reads,
\<
\label{HLUA}
&& H(w_1 , \dots , w_n | \bar{w}_1 , \dots , \bar{w}_n) = \nonumber \\
&&  (-1)^{L n + \frac{n(n+1)}{2}} c^{2 n} \frac{\displaystyle \prod_{j > i}^{n} b(w_i - w_j)^2 b(\bar{w}_i - \bar{w}_j)^2 a(w_j - \mu_i) a(\bar{w}_j - \mu_i)}{\prod_{i=1}^{n} b(w_i - \mu_i) b(\bar{w}_i - \mu_i)} \prod_{i=1}^{n} R_i^{-1} \Lambda_i \; , \nonumber \\
\>
with functions $R_i$ and $\Lambda_i$ given by
\<
\label{Ri}
R_i &=& \prod_{k=i}^{n} \frac{a(w_k - \mu_i)}{b(w_k - \mu_i)} - \prod_{k=i}^{n} \frac{a(\bar{w}_k - \mu_i)}{b(\bar{w}_k - \mu_i)} \nonumber \\
\Lambda_i &=& \prod_{k=i}^{L} a(\bar{w}_i - \mu_k) b(\mu_k - w_i) \prod_{k=i+1}^{n} \frac{a(w_i - w_k)}{b(w_i - w_k)} \frac{a(\bar{w}_k - \bar{w}_i)}{b(\bar{w}_k - \bar{w}_i)} \nonumber \\ 
&& - \prod_{k=i}^{L} a(w_i - \mu_k) b(\mu_k - \bar{w}_i) \prod_{k=i+1}^{n} \frac{a(w_k - w_i)}{b(w_k - w_i)} \frac{a(\bar{w}_i - \bar{w}_k)}{b(\bar{w}_i - \bar{w}_k)} \; .
\>
The formulae (\ref{LUA}) is commonly denominated off-shell scalar product as it is valid
for arbitrary complex parameters $\lambda_i^B$ and $\lambda_i^C$. 
In its turn, when the variables $\lambda_i^B$ are constrained by Bethe ansatz equations, see \cite{Korepin_book} for instance,
the function $S_n$ receives the name on-shell scalar product and the analysis of (\ref{SNAD}) in that case
has also been performed in \cite{Galleas_SC}.

\section{Partial differential equations}
\label{sec:PDE}

The functional equations (\ref{FZ}) and (\ref{SNAD}) contain a very rich structure which 
is not apparent at first sight. In order to illustrate how these hidden
structures emerge let us consider Eq. (\ref{FZ}) in the standard six-vertex model limit.
In that case we have $\mathcal{Z}_{\theta} \rightarrow Z$ and are left with the following equation,
\[
\label{FZT}
\bar{M}_0 \; Z(X^{1,L}) + \sum_{i=1}^{L} \bar{N}_i \; Z(X^{0,L}_i) = 0 \; ,
\]
with coefficients $\bar{M}_0$ and $\bar{N}_i$ given by
\<
\label{cfc}
\bar{M}_0 &=& \prod_{j=1}^{L} b(\lambda_0 - \mu_j) - \prod_{j=1}^{L} a(\lambda_0 - \mu_j) \prod_{j=1}^{L} \frac{a(\lambda_j - \lambda_0)}{b(\lambda_j - \lambda_0)} \nonumber \\
\bar{N}_i &=& \frac{c(\lambda_i - \lambda_0)}{b(\lambda_i - \lambda_0)} \prod_{j=1}^{L} a(\lambda_i - \mu_j) \prod_{\stackrel{j=1}{j \neq i}}^{L} \frac{a(\lambda_j - \lambda_i)}{b(\lambda_j - \lambda_i)} \; .
\>
In what follows we intend to demonstrate that (\ref{FZT}) encodes a family of partial differential
equations and for that we need to introduce some extra definitions and lemmas.

\begin{mydef} \label{Dia}
Let $f$ be a complex valued function $f(z) \in \mathbb{C}[z]$ and $z=(z_1, \dots , z_n) \in \mathbb{C}^n$.
Then for $\alpha \notin \{1, 2 , \dots, n \}$ we define the operator $D_i^{\alpha}$ as
\[
\label{dia}
D_i^{\alpha} \; : \qquad f(z_1 , \dots , z_i , \dots , z_n) \mapsto f(z_1 , \dots , z_{\alpha} , \dots , z_n)   \; .
\] 
The operator $D_i^{\alpha}$ essentially replaces the variable $z_i$ by $z_{\alpha}$. It is worth mentioning that
$D_i^{\alpha}$ had been previously introduced in \cite{Galleas11}.
\end{mydef}

\begin{lemma}[Differential realization] The operator $D_i^{\alpha}$ admits the realization
\[
\label{diff}
D_i^{\alpha} = \sum_{k=0}^{m} \frac{(z_{\alpha} - z_i)^k}{k!} \frac{\partial^k}{\partial z_i^k}
\]
when its action is restricted to the ring of multivariate polynomials of degree $m$.
\end{lemma}

\begin{proof}
Let $f = f(z_1 , z_2 , \dots , z_n)$ and $\mathbb{K}^m [z_1, z_2 , \dots , z_n]$ be the ring of polynomials in $z_1 , \dots , z_n$
with degree $m$. The ring $\mathbb{K}^m [z_1, z_2 , \dots , z_n]$ shall be simply denoted as $\mathbb{K}^m [z]$ and the condition
$f \in \mathbb{K}^m [z]$  implies 
\[
\label{dk}
\frac{\partial^k f}{\partial z_i^k} = 0 \qquad \mbox{if} \quad k > m  \; .
\]
Next we consider the Taylor expansion of $f$ in the variable $z_i$ around the point $z_{\alpha}$.
Due to (\ref{dk}) the expansion is truncated and convergent. Thus we have,
\<
\label{taylor}
f &=& f(\dots, z_{i-1} , z_{\alpha} , z_{i+1} , \dots) + \left. \frac{\partial f}{\partial z_i} \right|_{i = \alpha} (z_i - z_{\alpha}) \nonumber \\
&& + \frac{1}{2} \left. \frac{\partial^2 f}{\partial z_i^2} \right|_{i = \alpha} (z_i - z_{\alpha})^2 + \dots + \frac{1}{m!} \left. \frac{\partial^m f}{\partial z_i^m} \right|_{i = \alpha} (z_i - z_{\alpha})^m \; .
\>
The expression (\ref{taylor}) holds for indexes $i \in \{1 , 2, \dots , n \}$, and as long as $\alpha \notin \{1 , 2, \dots , n \}$
we can write
\[
\label{ai}
\left. \frac{\partial^k f}{\partial z_i^k} \right|_{i = \alpha} = \frac{\partial^k }{\partial z_{\alpha}^k} f(\dots, z_{i-1} , z_{\alpha} , z_{i+1} , \dots) \; .
\]
In this way formulas (\ref{taylor}) and (\ref{ai}) yields the following relation,
\[
\label{ff}
f(\dots, z_{i-1} , z_{i} , z_{i+1} , \dots) = \left[ \sum_{k=0}^{m} \frac{(z_i - z_{\alpha})^k}{k!} \frac{\partial^k }{\partial z_{\alpha}^k} \right] f(\dots, z_{i-1} , z_{\alpha} , z_{i+1} , \dots) \; .
\]
The term inside the brackets in (\ref{ff}) performs the operation (\ref{dia}) and we thus have the differential realization (\ref{diff}).
\end{proof}

Now the functional equation (\ref{FZT}) can be written in terms of operators $D_i^{\alpha}$. For that
we only need to consider $n=L$ and $z_i = \lambda_i$. By doing so (\ref{FZT}) becomes $\mathfrak{L}(\lambda_0) Z(X^{1,L}) = 0$
where
\[
\label{opl}
\mathfrak{L}(\lambda_0)  = \bar{M}_0 + \sum_{i=1}^{L} \bar{N}_i \; D_i^{0} \; .
\]
Some remarks are required at this stage. For instance, although the functional equation (\ref{FZT}) 
depends on the set of variables $X^{0,L}$, the use of Definition \ref{Dia} localizes the whole dependence with
the variable $\lambda_0$ in the operator $\mathfrak{L}$. It is also important to stress 
here that we can not immediately use the differential realization (\ref{diff}) in (\ref{opl}) since 
it is valid only for functions in $\mathbb{K}^m [z]$. Nevertheless, in what follows we shall discuss how (\ref{diff}) can
be adjusted for $Z(X^{1,L})$.

\begin{lemma}[Polynomial structure] \label{POL}
In terms of variables $x_i = e^{2 \lambda_i}$ the partition function $Z(X^{1,L})$
is of the form
\[
\label{pol}
Z(X^{1,L}) = \prod_{j=1}^{L} x_j^{\frac{1-L}{2}} \bar{Z}(x_1 , \dots , x_L) \; ,
\]
where $\bar{Z}(x_1 , \dots , x_L)$ is a polynomial of degree $L-1$ in each variable $x_i$
separately.
\end{lemma}
\begin{proof}
A detailed proof can be found in \cite{Korepin82} and \cite{Galleas10}.
\end{proof}

Lemma \ref{POL} is telling us that $Z(X^{1,L})$ consists of a multivariate polynomial
up to an overall multiplicative factor when the appropriate variable is considered. 
As a matter of fact we have $\bar{Z}(x_1 , \dots , x_L) \in \mathbb{K}^{L-1} [x]$ and therefore the realization 
(\ref{diff}) can be employed only for $\bar{Z}$. Due to that it is convenient to define the functions 
\[
\label{mnbar}
\check{M}_0 = \bar{M}_0 \prod_{j=1}^{L} x_j^{\frac{1-L}{2}} \qquad \mbox{and} \qquad \check{N}_i = \bar{N}_i \prod_{\stackrel{j=0}{j \neq i}}^{L} x_j^{\frac{1-L}{2}}
\]
in such a way that (\ref{FZT}) becomes $\bar{\mathfrak{L}}(x_0) \bar{Z}(x_1, \dots , x_L) = 0$ with
\[
\label{bopl}
\bar{\mathfrak{L}}(x_0)  = \check{M}_0 + \sum_{i=1}^{L} \check{N}_i \; D_i^{0} \; .
\]
Now the formulae (\ref{diff}) can be substituted into (\ref{bopl}) \footnote{Here we need to set $z_i = x_i$ due to the
change of variables discussed in Lemma \ref{POL}.} leaving us with the expression
\[
\bar{\mathfrak{L}}(x_0)  = \sum_{k=0}^{L-2} x_0^k \; \Omega_{k} \;\; .
\]
The coefficients $\Omega_k$ are differential operators whilst $\bar{\mathfrak{L}}(x_0)$ is a polynomial of
degree $L-2$ in the variable $x_0$. In this way the equation $\bar{\mathfrak{L}}(x_0) \bar{Z}(x_1, \dots , x_L) = 0$
needs to be independently satisfied for each power in $x_0$ which leaves us with the following family of 
partial differential equations,
\[
\label{omega}
\Omega_{k} \; \bar{Z}(x_1, \dots , x_L)  = 0 \qquad 0 \leq k \leq L-2 \; .
\]
Although the explicit form of the operators $\Omega_k$ can be written down for any length $L$, it is usually
given by cumbersome expressions for most of the indexes $k$. Fortunately the situation is different for $k=L-2$ 
and the leading term coefficient $\Omega_{L-2}$ exhibits a compact expression given by
\[
\label{om}
\Omega_{L-2} = \sum_{i=1}^{L} \bar{a}(x_i , y_i) - \frac{q^{2(1-L)}}{(L-1)!} \sum_{i=1}^{L} \prod_{j=1}^{L} \bar{a}(x_i , y_j) \prod_{\stackrel{j=1}{j \neq i}}^{L} \frac{\bar{a}(x_j , x_i)}{\bar{b}(x_j , x_i)} \frac{\partial^{L-1}}{\partial x_i^{L-1}}  \; .
\]
The expression (\ref{om}) takes into account the further conventions $q = e^{\gamma}$, $y_i = e^{2 \mu_i}$ and the 
remaining functions are then defined as $\bar{a}(x,y) = x q^2 - y$ and $\bar{b}(x,y) = x  - y$.

From (\ref{om}) we can see that the operator $\Omega_{L-2}$ exhibits some very interesting characteristics.
For instance, it naturally decomposes into two kinds of terms and it is tempting to interpret it as the 
hamiltonian of a many-body system. Although $\Omega_{L-2}$ contains higher order derivatives, the first term
in the RHS of (\ref{om}) could be thought of as `potential energy' while the second term can be regarded as
`kinetic energy'. The most obvious problem with this interpretation is that the interaction factor 
$\prod_{j \neq i}^{L} \frac{\bar{a}(x_j , x_i)}{\bar{b}(x_j , x_i)}$ appears in the `kinetic energy' term and it is
not clear if a change of variables could have this issue properly fixed. Nevertheless, taking into account this analogy 
it is sensible to consider the eigenvalue problem for the operator (\ref{om}), i.e. $\Omega_{L-2} \Psi = \Lambda \Psi$.
In this way the partition function $\bar{Z}$ can be regarded as the null eigenvalue wave function associated with
$\Omega_{L-2}$.

\section{Concluding remarks}
\label{sec:conclusion}

In this article we have described a mechanism allowing to extract functional equations
satisfied by certain partition functions of two-dimensional lattice models directly
from the Yang-Baxter algebra. More precisely, we have applied this method for the elliptic 
Eight-Vertex-SOS model with domain wall boundaries \cite{Galleas_2012} and for scalar products of
Bethe vectors \cite{Galleas_SC}. For those systems we have obtained functional relations 
satisfied by their partition functions whose solution are then given in terms of multiple 
contour integrals.

The class of functional equations we describe here share some similarities concerning
their structure with the classical Knizhnik-Zamolodchikov equation \cite{KZ_1984} as discussed in 
\Secref{sec:FEZ}. This similarity seems to extend to their solutions as multiple contour
integrals are also convenient to describe solutions of KZ equations \cite{Varchenko_book}. 

Although classical KZ equations consist of a system of partial differential equations,
whilst here we are dealing with functional equations, in \Secref{sec:PDE} we have also
demonstrated that there is a family of partial differential equations underlying
our functional relations obtained from the Yang-Baxter algebra.
Interestingly, one member of this family exhibits a structure which resembles that of a generalized
Schr\"odinger equation for a quantum many-body hamiltonian. 

Concerning this algebraic-functional approach, it is fair to say that this method
is still under development as there are still many opened questions. For instance, 
motivated by the similarities shared with the classical KZ equation, one could ask if 
there is an analogous of the whole theory of KZ equations \cite{Varchenko_book, Etingof_book}
for the functional equations presented here.
Moreover, as far as the computation of physical quantities are concerned,
one important problem is the evaluation of the model free-energy per site in
the thermodynamical limit from the multiple contour integrals presented here.
This analysis would not only give us access to the model physical properties
but also help us to understand the influence of boundary conditions for vertex
and SOS models \cite{Rosengren_2011, Justin_2000}.

\section{Acknowledgements}
\label{sec:ack}
The author is supported by the Netherlands Organization for Scientific
Research (NWO) under the VICI grant 680-47-602. The work of W. Galleas is also
part of the ERC Advanced grant research programme No. 246974, 
{\it ``Supersymmetry: a window to non-perturbative physics"}.



\bibliographystyle{hunsrt}
\bibliography{references}

\begin{thebibliography}{10}

\bibitem{Onsager_1944}
{Onsager, L.}
\newblock {Crystal statistics I. A two-dimensional model with an order-disorder
  transition}.
\newblock {\em {Phys. Rev.}}, {65}({3/4}):{117--149}, {1944}.

\bibitem{Ma}
S.K. Ma.
\newblock {\em Modern Theory of Critical Phenomena}.
\newblock Advanced book classics. Perseus, 2000.

\bibitem{Bethe_1931}
H.~Bethe.
\newblock {Zur Theorie der Metalle I. Eigenwerte und Eigenfunktionen der
  Linearen Atomkette}.
\newblock {\em {Zeitschrift f\"ur Physik}}, ({71}):{225--226}, {1931}.

\bibitem{Faddeev_1981}
{Faddeev, L. D. and Takhtajan, L. A.}
\newblock {What is the spin of a spin-wave?}
\newblock {\em {Phys. Lett. A}}, {85}({6-7}):{375--377}, {1981}.

\bibitem{Zamolodchikov_1990}
{Zamolodchikov, A. B.}
\newblock {Thermodynamic Bethe ansatz in relativistic models: Scaling $3$-state
  Potts and Lee-Yang models}.
\newblock {\em {Nucl. Phys. B}}, {342}({3}):{695--720}, {1990}.

\bibitem{Zarembo_2003}
J.~A. Minahan and K.~Zarembo.
\newblock {The Bethe-ansatz for $N=4$ super Yang-Mills}.
\newblock {\em {JHEP}}, ({3}), {2003}.

\bibitem{Galleas_Brak}
R.~Brak and W.~Galleas.
\newblock {Constant Term Solution for an Arbitrary Number of Osculating Lattice
  Paths}.
\newblock {\em {Lett. Math. Phys.}}, {103}({11}):{1261--1272}, {2013}.

\bibitem{Kramers_1941a}
H.~A. Kramers and G.~H. Wannier.
\newblock {Statistics of the two-dimensional ferromagnet Part I}.
\newblock {\em Phys. Rev.}, 60(3):252, 1941.

\bibitem{Kramers_1941b}
H.~A. Kramers and G.~H. Wannier.
\newblock {Statistics of the two-dimensional ferromagnet Part II}.
\newblock {\em Phys. Rev.}, 60(3):263, 1941.

\bibitem{Baxter_1971}
R.~J. Baxter.
\newblock Eight vertex model in lattice statistics.
\newblock {\em Phys. Rev. Lett.}, 26:832, 1971.

\bibitem{Sk_1979}
E.~K. Sklyanin, L.~A. Takhtadzhyan, and L.~D. Faddeev.
\newblock {Quantum Inverse Method.1}.
\newblock {\em Theor. Math. Phys.}, 40:688, 1979.

\bibitem{Fad_1979}
L.~A. Takhtadzhyan and L.~D. Faddeev.
\newblock {The quantum method of the inverse problem and the Heisenberg $XYZ$
  model}.
\newblock {\em Russ. Math. Surv.}, 34:11, 1979.

\bibitem{chari1995guide}
V.~Chari and A.N. Pressley.
\newblock {\em A Guide to Quantum Groups}.
\newblock Cambridge University Press, 1995.

\bibitem{KZ_1984}
{Knizhnik, V. G. and Zamolodchikov, A. B.}
\newblock {Current algebra and Wess-Zumino model in $2$ dimensions}.
\newblock {\em {Nucl. Phys. B}}, {247}({1}):{83--103}, {1984}.

\bibitem{Babujian_1993}
{Babujian, H. M.}
\newblock {Off-shell Bethe ansatz equations and $N$-point correlators in the
  $SU(2)$ WZNW theory}.
\newblock {\em {J. Phys. A: Math. and Gen.}}, {26}:{6981--6990}, {1993}.

\bibitem{Babujian_1994}
{Babujian, H. M. and Flume, R.}
\newblock {Off-shell Bethe ansatz equation for Gaudin magnets and solutions of
  Knizhnik-Zamolodchikov equations}.
\newblock {\em {Mod. Phys. Lett. A}}, {9}({22}):{2029--2039}, {1994}.

\bibitem{Galleas_2008}
W.~Galleas.
\newblock {Functional relations from the Yang-Baxter algebra: Eigenvalues of
  the $XXZ$ model with non-diagonal twisted and open boundary conditions}.
\newblock {\em {Nucl. Phys. B}}, {790}({3}):{524--542}, {2008}.

\bibitem{Galleas10}
W.~Galleas.
\newblock Functional relations for the six-vertex model with domain wall
  boundary conditions.
\newblock {\em J. Stat. Mech.}, (06):P06008, 2010.

\bibitem{Galleas11}
W.~Galleas.
\newblock A new representation for the partition function of the six-vertex
  model with domain wall boundaries.
\newblock {\em J. Stat. Mech.}, (01):P01013, 2011.

\bibitem{Galleas_2011}
W.~Galleas.
\newblock {Multiple integral representation for the trigonometric SOS model
  with domain wall boundaries}.
\newblock {\em {Nucl. Phys. B}}, {858}({1}):{117--141}, {2012},
  {math-ph/1111.6683}.

\bibitem{Galleas_2012}
W.~Galleas.
\newblock {Refined functional relations for the elliptic SOS model}.
\newblock {\em {Nucl. Phys. B}}, {867}:{855--871}, {2013}.

\bibitem{Galleas_SC}
{Galleas, W.}
\newblock {Scalar product of Bethe vectors from functional equations}.
\newblock {arXiv: 1211.7342}.

\bibitem{Baxter_book}
R.~J. Baxter.
\newblock {\em {Exactly Solved Models in Statistical Mechanics}}.
\newblock Dover Publications, Inc., Mineola, New York, 2007.

\bibitem{Smirnov_2001}
S.~Smirnov.
\newblock {Critical percolation in the plane: conformal invariance, Cardy's
  formula, scaling limits}.
\newblock {\em {Comptes rendus de L'academie des sciences serie I -
  Mathematique}}, {333}({3}):{239--244}, {2001}.

\bibitem{Smirnov_pre}
S.~Smirnov.
\newblock {Critical percolation in the plane. I. Conformal invariance and
  Cardy's formula. II. Continuum scaling limit.}
\newblock {\em {Pre-print}}, {2001}, {arXiv: math.PR/1211.3968}.

\bibitem{Cardy_1992}
J.~L. Cardy.
\newblock {Critical percolation in finite geometries}.
\newblock {\em {J. Phys. A - Math. and Gen.}}, {25}({4}):{L201--L206}, {1992}.

\bibitem{Felder_1994}
G.~Felder.
\newblock Conformal field theory and integrable systems associated to elliptic
  curves.
\newblock {\em Proceedings of the International Congress of Mathematicians},
  1:1247, 1995.

\bibitem{Felder_1996}
G.~Felder.
\newblock Algebraic bethe ansatz for the elliptic quantum group $e_{\tau ,
  \eta} (sl_2)$.
\newblock {\em Nucl. Phys. B}, 480:485, 1996.

\bibitem{Watson}
E.~T. Whittaker and G.~N Watson.
\newblock {\em A Course of Modern Analysis}.
\newblock Cambridge University Press, fourth edition, 1927.

\bibitem{integrable_book}
G.~M. D'Ariano, A.~Montorsi, and M.~G. Rasetti.
\newblock {\em Integrable Systems in Statistical Mechanics}.
\newblock World Scientific, 1985.

\bibitem{Rosengren_2008}
H.~Rosengren.
\newblock {An Izergin-Korepin type identity for the 8VSOS model with
  applications to alternating sign matrices}.
\newblock {\em Adv. Appl. Math.}, 43:137, 2009.

\bibitem{Pakuliak_2008}
S.~Pakuliak, V.~Rubtsov, and A.~Silantyev.
\newblock {SOS model partition function and the elliptic weight function}.
\newblock {\em J. Phys. A}, 41:295204, 2008.

\bibitem{WengLi_2009}
W.-L. Yang and Y.-Z. Zhang.
\newblock {Partition function of the eight-vertex model with domain wall
  boundary condition}.
\newblock {\em J. Math. Phys.}, 50:083518, 2009.

\bibitem{Razumov_2009a}
A.~G. Razumov and Y.~G. Stroganov.
\newblock {Three-coloring statistical model with domain wall boundary
  conditions: Functional equations}.
\newblock {\em Theor. Math. Phys.}, 161:1325, 2009.

\bibitem{Korepin82}
V.~E. Korepin.
\newblock {Calculation of norms of wave functions}.
\newblock {\em Commun. Math. Phys.}, 86:391--418, 1982.

\bibitem{deGier_Galleas11}
J.~de~Gier, W.~Galleas, and M.~Sorrell.
\newblock Multiple integral formula for the off-shell six vertex scalar
  product.
\newblock 2011, hep-th/1111.3712.

\bibitem{Korepin_book}
V.~E. Korepin, N.~M. Bogoliubov, and A.~G. Izergin.
\newblock {\em Quantum inverse scattering method and correlation functions}.
\newblock Cambridge University Press, 1993.

\bibitem{Aczel_book}
{Acz\'el, J.}
\newblock {\em {Functional equations: History, Applications and Theory}}.
\newblock {D. Reidel Publishing Company}, 1984.

\bibitem{Bazhanov_1997}
V.~V. Bazhanov, S.~L. Lukyanov, and A.~B. Zamolodchikov.
\newblock {Integrable structure of conformal field theory - II. Q-operator and
  DDV equation}.
\newblock {\em {Comm. Math. Phys.}}, {190}({2}):{247--278}, {1997}.

\bibitem{Baxter_1972}
R.~J. Baxter.
\newblock {Partition-function of 8-vertex lattice model}.
\newblock {\em {Ann. Phys.}}, {70}({1}):{193}, {1972}.

\bibitem{Stroganov_1982}
Y.~G. Stroganov.
\newblock {General properties and particular solutions of the triangle
  equation. Calculation of the partition function for some models on the plane
  lattice}.
\newblock {\em Unpublished thesis}, 1982.

\bibitem{Reshet_1987}
N.~Y. Reshetikhin.
\newblock {The spectrum of the transfer-matrices connected with Kac-Moody
  algebras}.
\newblock {\em {Lett. Math. Phys.}}, {14}({3}):{235--246}, {1987}.

\bibitem{Kuniba_1994}
A.~Kuniba, T.~Nakanishi, and J.~Suzuki.
\newblock {Functional relations in solvable lattice models. 1. Functional
  relations and representation theory}.
\newblock {\em {Int. J. Mod. Phys. A}}, {9}({30}):{5215--5266}, {1994}.

\bibitem{Kuniba_2011}
A.~Kuniba, T.~Nakanishi, and J.~Suzuki.
\newblock {$T$-systems and $Y$-systems in integrable systems}.
\newblock {\em {J. Phys. A - Math. and Theor.}}, {44}({10}), {2011}.

\bibitem{Frenkel_Reshetikhin}
I.~B. Frenkel and N.~Y. Reshetikhin.
\newblock {Quantum affine algebras and holonomic difference equations}.
\newblock {\em {Comm. Math. Phys.}}, {146}({1}):{1--60}, {1992}.

\bibitem{Smirnov_book}
{Smirnov, F. A.}
\newblock {\em {Form factors in completely integrable models of quantum field
  theory}}.
\newblock {Advanced Series in Mathematical Physics. 14. Singapore: World
  Scientific,. xi, 208 p. }, {1992}.

\bibitem{Varchenko_book}
{Varchenko, A. N.}
\newblock {\em {Special Functions, KZ Type Equations, and Representation
  Theory}}.
\newblock {American Mathematical Society}, {2003}.

\bibitem{Weber}
H.~Weber.
\newblock {\em Elliptische Functionen und algebraische Zahlen}.
\newblock Friedrich Vieweg und Sohn, Braunschweig, fourth edition, 1891.

\bibitem{Etingof_book}
{Etingof, P. I. and Frenkel, I. B. and Kirillov, A. A. }.
\newblock {\em {Lectures on Representation Theory and Knizhnik-Zamolodchikov
  Equations}}.
\newblock {American Mathematical Society}, {1998}.

\bibitem{Rosengren_2011}
H.~Rosengren.
\newblock {The three-colour model with domain wall boundary conditions}.
\newblock {\em Adv. Appl. Math.}, 46:481, 2011.

\bibitem{Justin_2000}
V.~E. Korepin and P.~Zinn-Justin.
\newblock Thermodynamic limit of the six-vertex model with domain wall boundary
  conditions.
\newblock {\em J. Phys. A: Math. Gen.}, 33:7053, 2000.

\end{thebibliography}

\end{document}